\documentclass[sigconf]{acmart}
\usepackage{amsmath} % recommended

\usepackage{array}
\usepackage{makecell}

\usepackage[english]{babel}
\newtheorem{theorem}{Theorem}[section]

\newtheorem{Assumption}[theorem]{Assumption}

\AtBeginDocument{%
  \providecommand\BibTeX{{%
    \normalfont B\kern-0.5em{\scshape i\kern-0.25em b}\kern-0.8em\TeX}}}

\setcopyright{acmcopyright}
\copyrightyear{2024}
\acmYear{2024}
\acmDOI{XXXXXXX.XXXXXXX}

\acmConference[ACM e-Energy 2024]{ACM e-Energy 2024}{June 04--07,
  2024}{Singapore}
\acmPrice{15.00}
\acmISBN{978-1-4503-XXXX-X/18/06}

\begin{document}

\title{Contributions of Individual Generators to \\ Nodal Carbon
  Emissions}

\author{Yize Chen}
\email{yizechen@ust.hk}
\affiliation{%
  \institution{Hong Kong University of Science and Technology} \city{Hong Kong SAR} \country{China}
}
\author{Deepjyoti Deka}
\email{deepjyoti@lanl.gov}
\affiliation{%
  \institution{Los Alamos National Laboratory}
  \city{Los Alamos}
  \state{NM}
  \country{USA}
}
\author{Yuanyuan Shi}
\email{yyshi@eng.ucsd.edu}
\affiliation{%
  \institution{University of California San Diego}
  \city{San Diego}
  \state{CA}
  \country{USA}
}

\begin{abstract}
Recent shifts toward sustainable energy systems have witnessed the fast deployment of carbon-free and carbon-efficient generations across the power networks. However, the benefits of carbon reduction are not experienced evenly throughout the grid. Each generator can have distinct carbon emission rates. Due to the existence of physical power flows, nodal power consumption is met by a combination of a set of generators, while such combination is determined by network topology, generators' characteristics and power demand. This paper describes a technique based on physical power flow model, which can efficiently compute the nodal carbon emissions contributed by each single generator given the generation and power flow information. We also extend the technique to calculate both the nodal average carbon emission and marginal carbon emission rates. Simulation results validate the effectiveness of the calculations, while our technique provides a fundamental tool for applications such as carbon auditing, carbon-oriented demand management and future carbon-oriented capacity expansion.

  %Environmental concerns associated with power systems drive an increasing interest in finding sustainable electricity power generation planning and operation schemes. To better facilitate the carbon reduction goals, profiling system-level and node-level carbon emissions serve as a cornerstone. 
\end{abstract}

\maketitle

\section{Introduction}
Mitigating climate change has emerged as an essential task for modern electric power systems via reliably generating and delivering low-carbon power to customers. Measuring and evaluating carbon emissions has recently become an integral process en route to decarbonizing energy systems, drawing widespread attention from a full spectrum of industries~\cite{anthony2020carbontracker, dixon2020scheduling, jenkins2009interventions} and policymakers~\cite{nelson2012high, liu2022challenges, rudkevich2011locational}. Given carbon emission rates, various applications have been discussed to reduce system carbon emissions via data center scheduling \cite{lin2023adapting, lindberg2022using}, demand-side management~\cite{park2023decarbonizing, wang2021optimal}, energy storage and electric vehicle load shifting~\cite{cheng2022carbon}, and power grid resource planning~\cite{abdennadher2022carbon,chen2023carbon}. 

To implement carbon reduction actions and plan for a more sustainable energy system, quantifying and auditing the fine-grained carbon emission rate serves as a cornerstone. In current practice, the \emph{system-level carbon emissions} are calculated by summing up all generators' carbon emissions~\cite{hundiwale2016greenhouse}. For instance, ``virtual'' carbon flow has been proposed to measure the transfer of carbon emission between different geographical regions~\cite{elkins2001carbon}. And pure \emph{statistical methods} such as \cite{leerbeck2020short, lau2014modelling} focus on short-term carbon emission forecasting task using statistical features including weather conditions and load. However, these methods output the whole network's marginal and average carbon emissions rather than nodal ones. 

Tools that can quantify nodal-level carbon emissions are becoming increasingly important. 
Such tools can provide system operators and energy users with real-time carbon emission information~\cite{park2023decarbonizing,cheng2019low}, and guide carbon reduction plans by identifying each generator's emission profiles~\cite{van2012assessing}.
% as it can help analyze grid carbon emission patterns~\cite{park2023decarbonizing}, guide and inform energy users on the real-time carbon emission information~\cite{cheng2019low}, and  set carbon reduction plans by identifying each generator's emission profiles~\cite{van2012assessing}. 
Both the nodal \emph{average carbon emission} and \emph{marginal carbon emission} (or locational marginal emission) rates are useful metrics, where the average value quantifies the overall carbon emission rate of nodal power consumption, while the latter reflects the sensitivity with respect to nodal power demand. %Both values are useful metrics, as they can be applied to either represent/audit each region's carbon emission profiles, or measure marginal emission effects of electricity demand at a particular node.

Due to the network structure, time-varying loads, and complex physical constraints in power gird, tools for analyzing nodal carbon emissions are still under investigation. By using a lookup table of marginal emission rate versus load value, a load control strategy is established in \cite{wang2014locational} for carbon reduction. The load shifting problem in \cite{lindberg2022using, lindberg2021guide} considers carbon reduction by quantifying nodal marginal emission rate, while it is limited to a neighborhood load region with respect to original load vector. Techniques based on implicit function theorem is developed in \cite{valenzuela2023dynamic} to calculate the Jacobian and associated solution map from demand vector to generator vector, thereafter the marginal carbon emission is calculated. As for nodal average carbon emission, \cite{li2013carbon, kang2015carbon} propose analytical models of carbon emission flow, yet the proposed iterative algorithm needs to update the estimated carbon emissions from respective generators with no convergence guarantees. Specifically, it calculates each node's power source mix, and resorts to relatively expensive matrix inversion to find the carbon emission mapping from generators to demands. Yet such matrix is not guaranteed to be invertible~\cite{cheng2005ptdf}, which can lead to unsolvable cases of carbon emission rates. Recent work incorporates such calculation into an optimal power flow framework with carbon emission constraints~\cite{chen2023carbon}. Our technique differentiates from \cite{kang2015carbon,chen2023carbon} by directly finding individual generator's contribution to each line and load, which is inversion-free and computationally efficient via depth-first search. Moreover, our algorithm can be readily used for estimating both the average and locational marginal emissions.

In this work, we propose an efficient algorithm which can calculate the exact nodal average carbon emission and marginal carbon emission rates. We focus on the carbon flow physically coupled with power flow, and develops a recursive algorithm to trace back each generator's carbon and power contribution with respect to each line flow and node. Interestingly, the resulting framework adapts depth-first tree search to find each generator's reachable sets with mild computation burden. The algorithm is network agnostic, and can be applied to different load conditions.  Simulation results validate the algorithm's efficiency in both determining each generator's contributions with respect to any node and quantifying nodal carbon emissions, so that policymakers, system operators and customers can better analyze the grid's carbon emission profiles. To facilitate future carbon-oriented task developments, we make our code publicly available at \url{https://github.com/chennnnnyize/Carbon_Emission_Power_Grids}.

%Estimated solution and analytical solution.

%Therefore, the following question arises: \emph{How can we decide the nodal average and marginal carbon emission rates?} This is a practical problem to which industry have not got clear answers yet.

\section{Methodology}
\subsection{Power Dispatch Model}
We consider a connected power network with $N$ nodes and $L$ power lines. Let $\mathcal{N}$ be the set of all nodes, $\mathcal{L}$ the set of lines and $\mathcal{G}$ the set of generators. At each timestep, given load vector  $\mathbf{p}^d \in \mathbb{R}^N$, we assume system operator solves the electricity market dispatch problem and finds the power dispatch $\mathbf{p}^g \in \mathbb{R}^K$ for $K$ generators. The corresponding line flow $\mathbf{p}_{line}\in \mathbb{R}^L$ is also solved. Without loss of generality, for line pair $(i,j)\in \mathcal{L}$, we use $p_{ij}$ to denote the line flow from node $i$ to node $j$. We refer to Appendix \ref{sec:opf} for the dispatch model of DC power flow in details. We use $\mathcal{N}^+_i, \mathcal{N}^-_i$ to denote the set of neighbor nodes that send power to and receive power from node $i$, respectively.

This paper aims at finding the carbon emissions associated with each load node. Mathematically, let $\gamma_k$ denote the carbon emissions rate of generator $k$. We want to recover the locational average carbon emission rate $\delta(p_i^d)$ and marginal rate $\mu(p_i^d)$. Both rates have a unit of lbs CO$_2$/MWh. 

To achieve this goal, we find it possible to follow the physical interpretations of average carbon emission rate $\delta(p_i^d)$ by taking the division of nodal total carbon emission by power demand $p_i^d$. The key for computing the total carbon emission is to find generator $k$'s contribution $p_i^d(p_k^g)$ to supply $p_i^d$ in MW\footnote{Throughout the paper, we use $p_i(x)$ to denote the power contributed by $x$ to node $i$, where $x$ can be generator or inflow. Similar definitions hold for line flow $p_{ij}(x)$}, then we can compute the total emission as $e(p_i^d)=\sum_{k=1}^{K} \gamma_k \cdot p_i^d(p_k^g)$. Thus, each node's average carbon emission rate (with $p_i^d>0$) is computed as,
\begin{equation}
\label{equ:emission}
    \delta(p_i^d)=\frac{\sum_{k=1}^{K} \gamma_k \cdot p_i^d(p_k^g)}{p_i^d}.
\end{equation}

\subsection{Tracing Individual Generator's Carbon Flow}
\label{sec:method}
Now we discuss how to find such generator $k$'s contribution $p_i^d(p_k^g)$ to node  $i$'s demand for any $(k,i)$ pair in the network.  Let $p_{ij}(p_k^g)$ denote the power flow on the line between node $i$ and node $j$ due to generator $k$. To realize the calculation of $p_i^d(p_k^g)$, we utilize the fact that it is possible to trace every line's power contribution $p_{ij}(p_k^g)$. After that, we can calculate each node's ``inflow mix'', which can be further tracked back to each node's generation mix.

\begin{figure}[]
\centering
\includegraphics[width=0.48\textwidth]{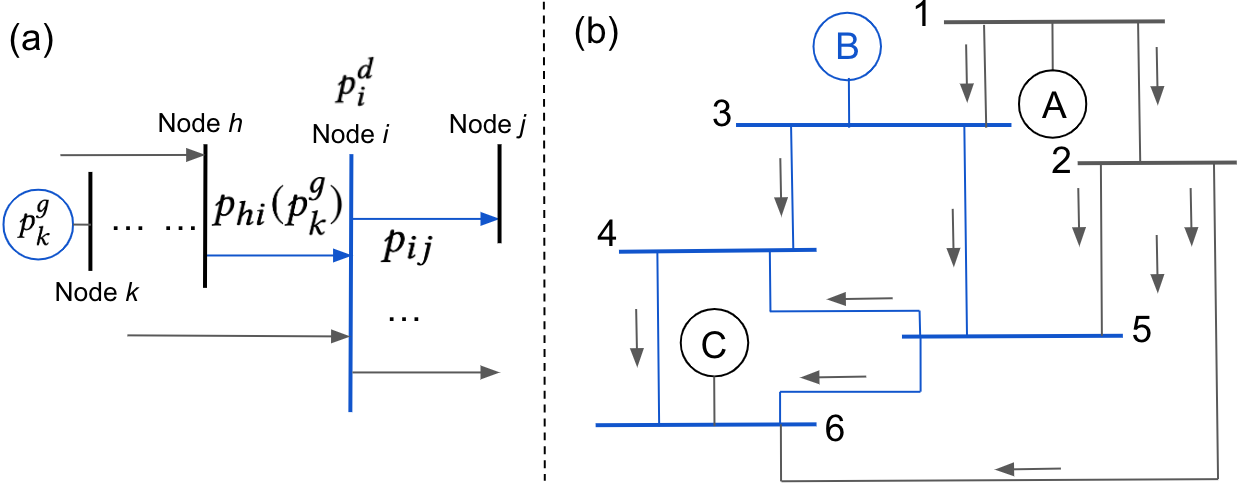}
\caption{(a). The overview of calculating generator's contribution at each node and line flow; (b). 6-bus example. The reachable path and nodes for generator $B$ is marked in blue.}
\label{fig:flow_example}
\end{figure}

In the following, we present an assumption on the proportional allocation of power in the network:
\begin{Assumption}
\label{assumption}
For any node $i$, if the proportion of the inflow which can be traced to generator $k$ is $\alpha_{i}(p_k^g)$, then the proportion of the outflow which can be traced to generator $k$ is also $\alpha_i(p_k^g)$.
\end{Assumption}

This assumption guides the proportional share of generators' generated power at each node. Essentially, the physically consumed power does not belong to any generator, while Assumption \ref{assumption} provides a principle to proportionally allocate the inflow power by nodal demand and output power flow. As illustrated in Fig. \ref{fig:flow_example}(a), we analyze the contribution made by generator $k$ to node $i$ via line $hi$. For any node in the network, provided we know $p_{hi}(p_k^g)$, we can explicitly write out the share of inflow power $p_{hi}(p_k^g)$ on nodal demand and outflow as 
\begin{subequations}
\label{equ:share}
\begin{align}
            p_i^d(p_{hi}(p_k^g))=& p_{hi}(p_k^g) \cdot \frac{p_i^d}{{p_i^d+\sum_{j\in \mathcal{N}^-_i}}p_{ij}};\\
        p_{ij}(p_{hi}(p_k^g))=& p_{hi}(p_k^g) \cdot \frac{p_{ij}}{{p_i^d+\sum_{j\in \mathcal{N}^-_i}}p_{ij}}.
\end{align}
\end{subequations}

The denominators in Equ \eqref{equ:share} are nonnegative as long as the load or line flow are non-negative, making the calculation of nodal power share always valid. An illustration of such calculation is shown in Fig. \ref{fig:6bus}(a), with elements marked in blue denoting the calculations involved in \eqref{equ:share}. By summing up all inflow's share from generator $g$, we can then calculate each node's power demand contributed by generator $g$ as  
\begin{equation}
\label{equ:summation}
    p_i^d(p_k^g)=\sum_{h\in \mathcal{N}^+_i} p_i^d(p_{hi}(p_k^g)).
\end{equation}

% By implementing the nodal line flow proportional share Equ \eqref{equ:share} and nodal generator contribution summation Equ \eqref{equ:summation}, we are readily resort to Equ \eqref{equ:emission} for finalizing the calculation of $\delta(p_i^d)$. 
By implementing the nodal line flow proportional share from Equ \eqref{equ:share} and summing the line flows with Equ \eqref{equ:summation} for each generator $k$, we are ready to use Equ \eqref{equ:emission} for computing $\delta(p_i^d)$.
Now the remaining challenge lies on how we can efficiently find each line's power flow  $p_{hi}$ contributed by generator $k$. It turns out as long as the line is reachable by generator $k$, e.g., there is a directed path from node $k$ to line $hi$, then line $hi$ shares a proportion of power provided by this generator. Starting from each generator, we can thus implement Equ \eqref{equ:share} recursively to find $p_{hi}(p_k^g)$. We find the following observation useful for the implementation of our carbon tracing algorithm.

%\textcolor{red}{Deep: proposition is about nodes but you want to calculate contribution to lines from generators...maybe add a line how that follows}
\begin{proposition}
\label{prop:reachable}
A reachable set of generator $k$ is defined as $\mathcal{S}_k:=\{i \in \mathcal{N}|p_i^d(p_{k}^g)\neq 0\}$. Given network demand $\mathbf{p}^d$ and generation $\mathbf{p}^g$, for any node $i$, it  belongs to at least one reachable set $S_k, \; g \in \mathcal{G}$.
\end{proposition}

\begin{proof}
We prove it by contradiction. If for node $i$, $i\notin \mathcal{S}_k, \forall k\in \mathcal{G}$, then by the definition of reachable set, we have $p_{i}^d(p_k^g)=0$ for all $k \in \mathcal{G}$, which means no injected power at node $i$. Then the power balance does not hold for node $i$, which is a contradiction. 
% \todo{we have $p_{hi}(p_k^g)=0, \forall h\in \mathcal{N}^+_i, \forall k \in \mathcal{G}$ (I think by the definition of reachable set, we have $p_{i}^d(p_k^g)=0$ for all $k \in \mathcal{G}$ rather than the line term)}, and no generation is located at $i$. Then the power balance does not hold for node $i$, which is a contradiction. 
\end{proof}

\begin{figure*}[]
\centering
\includegraphics[width=1.0\textwidth]{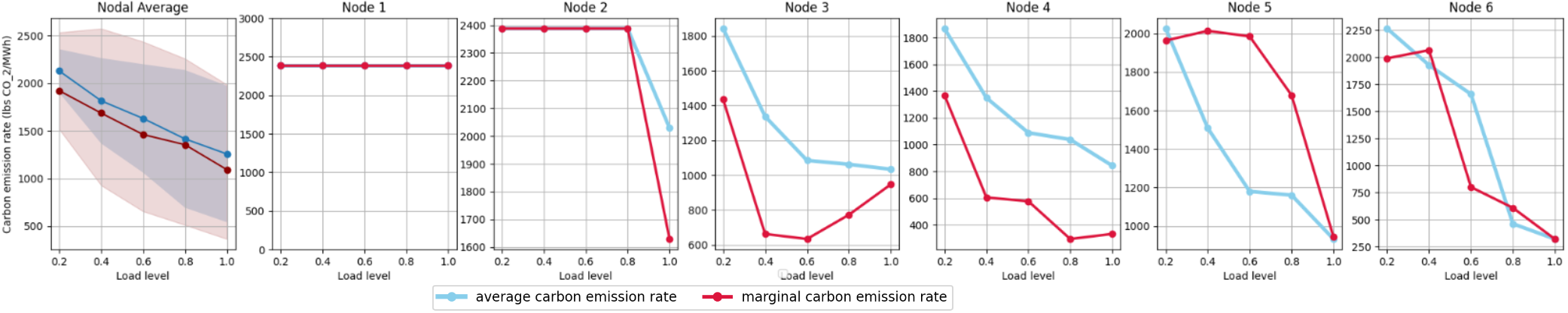}
\caption{The marginal and average carbon emission rate for each node in 6-bus example. The leftmost figure illustrate the mean and variance of average and marginal carbon emission rate.  }
\label{fig:6bus}
\end{figure*}

Based on Proposition \ref{prop:reachable}, the task boils down to constructing the reachable set for every generator with nonzero generations. We now show how we use directed line flow $p_{ij}$ to connect such set as well as tracing generator's contribution. Essentially, such reachable set for generator $k$ can be represented as a tree structure rooted at node $k$. This tree is a directed acyclic graph, where the path direction is determined by line flow direction. We note that due to the physics of power flow, there exists no directed cycle in the graph, thus the algorithm presented in this section will always have termination point. The proof is described in Appendix \ref{sec:cycle}.

For our tree-search algorithm, the termination is reached when all the reachable nodes are visited by the path belonging to the tree. An illustrative example is given in Fig. \ref{fig:flow_example}, where we mark the reachable set for generator $B$ in blue. To construct the tree, we start from node 3, find the directed line flow $p_{34}$ and $p_{35}$ first. It is then possible to compute the generation power shared by line 34 and line 35. Similarly, at node $i$ reached by inflow line $hi$ (node 5 and node 6 in this case), we identify the outflow line $ij, \forall j \in \mathcal{N}^-$ ($p_{54}$, $p_{56}$ and $p_{46}$). The algorithm continues working on newly added lines until all such lines are identified. 

Such tree structure can be retrieved by depth-first search. More importantly, once every node is reached during the search process by line $hi$, we can use Equ \eqref{equ:share} to calculate the power share provided by the inflow $p_{hi}$. By recursively implementing the power sharing calculation along with the tree search, once we traverse all nodes and edges in the directed tree for generator $k$, we know the exact contribution of this generator to each node (and each line flow). We implement such procedures for all generators. Then, for each node $i$, we can sum up all generators' carbon emission contributions and get total emission $e(p_i^d)$. For each tree, the complexity is $O(|N|+|L|)$ for the depth-first search, and the additional complexity comes from calculating each node and line flow's power share via Equ \eqref{equ:share}, and the worst case is $O(|N||L|)$. Combined together, this is still more efficient than previous approaches with matrix inversion involved~\cite{kang2015carbon, li2013carbon}. In addition, the proposed approach can be implemented in parallel by indexing and sharing the paths and nodes in all generators' trees. %Our algorithm can not only get locational average and marginal carbon emissions, but also we can find each generator's contribution with respect to loads.

Our approach is inspired by \cite{kirschen1997contributions}, where the authors proposed a novel method for identifying the common node clusters given line flow patterns, while each cluster shares the same group of generators' inputs. Our method differs from \cite{kirschen1997contributions} in directly working on single generator's reachable set, which avoids the computation for finding clusters. We also bridge the contribution of generation power and carbon emissions, and find a  numerical way to compute the marginal carbon emission rate in the next subsection.

We note that the directed trees can be different under different load conditions. For instance, in Fig. \ref{fig:flow_example}(b), if generator A has the lowest generation cost, generator A will supply all nodes in the graph when total load is small. While due to line flow limits and generation limits, Generator B and C will start to serve more nodes when load increases. Thus it is important to profile the line flow directions and construct each generator's reachable set.

\subsection{From Average to Marginal Carbon Emission }
\label{sec:marginal}
Mathematically, locational marginal emission rates can be derived by first calculating how generation changes with respect to nodal demand change, and then multiplying the change in generation by the emission rates of each generator. In our framework, we can implement a sensitivity analysis to estimate such marginal rate for each node. We calculate the nodal total emission before and after a small demand value change $\epsilon$ at node $i$, while keeping loads at all other nodes fixed. The locational marginal emission at node $i$ can be calculated as %\todo{Using \eqref{equ:marginal}, do we need to keep other nodes' load fix when perturbing node $i$'s demand $p_i^d$? I added the previous short sentence}
\begin{equation}
\label{equ:marginal}
    \mu(p_i^d)=\frac{d e(p_i^d)}{d p_i^d} \approx \frac{e(p_i^{d+\epsilon})-e(p_i^d)}{\epsilon}.
\end{equation}
During implementation, we get $e(p_i^{d+\epsilon})$ by running the full algorithm to update the carbon contribution by each generator. We note that such calculation based on load perturbation is exact for almost all load vectors. This is due to the fact that in the underlying power dispatch model, there are particular load regions which share the same power dispatch rules. 
%which decides generator's nodal carbon emission contribution. 
For a small $\epsilon$, as long as $p_i^{d+\epsilon}$ lies in the same load region as $p_i^d$, our calculation \eqref{equ:marginal} gives the exact marginal carbon emission rate. Such characteristic has been applied to locational marginal price analysis, and we refer more discussions in Appendix. \ref{sec:region} and mathematical details to \cite{ji2016probabilistic}.

The following observation is instantly made with respect to a group of nodes. Utilizing the observation can accelerate the carbon emission calculation by skipping such nodes. 

\begin{proposition}
For any node $i$ with no inflow , e.g., $p_{hi}=0, \forall h\in \mathcal{N}^+_i $, then we have $\delta(p_i^d)=\mu(p_i^d)=\gamma_i$, where $\gamma_i$ is the carbon emission rate of generator located at node $i$. 
\end{proposition}

 \section{Illustrative Example}
We test on IEEE 6-bus (in Fig.~\ref{fig:flow_example}) and 30-bus system (in Fig.~\ref{fig:30bus_map}) to validate our carbon emission calculation framework. We use the 6-bus example to quantify how load change will affect both the average and marginal carbon emissions. For the 30-bus system, we share some observations about the overall system emissions and nodal patterns. To enable the carbon emission analysis, we first solve the power dispatch problem to find the optimal generation dispatch solutions and line power flows. The IEEE 30-bus system with 6 generators~\cite{shahidehpour2003appendix} represents a portion of the Midwestern United States Electric Power System.  We utilize the generator data collected for the U.S. Midwest Reliability Organization (MRO) region in \cite{deetjen2019reduced}, and get realistic generators' emission rate, generation fuel type, and annual generation data. We simulate the case where only fossil fuel generators are used to power the grid. The generators' carbon emission rates are in the range of $[113, \; 2388]$ lbs CO$_2$/MWh, and in general, generators with higher generation costs (e.g., gas and nuclear) correspond to lower carbon emission rate in our simulating instances. For both systems, we firstly identify a load condition which approaches system limits. Then we vary the load vectors by taking $20\%, 40\%, 60\%, 80\%$ of such load, and calculate the average and marginal carbon emission rates respectively. More simulation setting details are described in Appendix \ref{sec:setting}.%\footnote{\url{https://labs.ece.uw.edu/pstca/pf30/ieee30cdf.txt}}.

\begin{figure}[]
\centering
\includegraphics[width=0.4\textwidth]{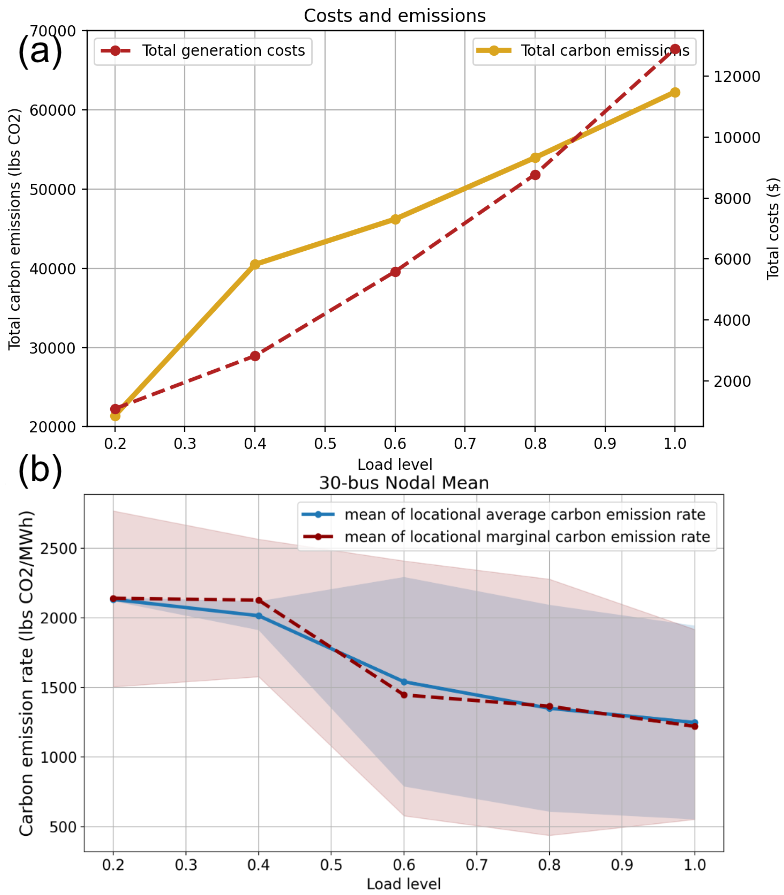}
\caption{IEEE 30-bus system simulation results. (a). Total carbon emissions and generation costs; (b). Mean of all nodes' average carbon emissions and marginal carbon emissions.}
\label{fig:30bus}
\end{figure}

In Fig. \ref{fig:6bus}, we show the average and marginal carbon emissions for each node, along with the mean of all nodes' carbon emission metrics. Each node's carbon emission rate differs a lot, with Node 1 having the highest emission rates in terms of both the average and marginal emissions, since this load is always supplied by the generator attached on the same bus (see Fig. \ref{fig:flow_example}(b)), which is assigned the highest generator's carbon emission rate. Node 6 reaches the smallest carbon emission rate when the load level is high, with both average and marginal rates under 500 lbs CO$_2$/MWh. For all nodes, the average emission rate is decreasing with the increase of load level. This is because generators with lower emission rates are contributing more to meet the higher demands. An interesting case is for Node 2, where power is predominately contributed by the generator attached at node 1 when load is in $[20\%, 80\%]$ variation range. But both the locational average and locational marginal emission rates drop when load reach to $100\%$. This is caused by the reversed line flow from $p_{25}$ and $p_{26}$ to $p_{52}$ and $p_{62}$, which switches part of Node 2's load contributor. There is not a clear pattern on whether the average rate is higher than the marginal rate, or the other way around. Both rates are affected by generator's contribution with respect to the given loads and line flow conditions.

Furthermore, we calculate the mean and standard deviation of all nodes' carbon emission rate, and show in the leftmost subplot of Fig. \ref{fig:6bus}. It can be observed that all nodes' emission rate are decreasing when load becomes larger. 
The average emission rate is consistently larger than the marginal emission rate, possibly due to the fact that generators with lower emission rate start to contribute (as the marginal generators) when we reach higher load level. For both carbon metrics, we notice the standard deviation is quite large, indicating huge differences for the nodal carbon emission patterns. 

In Fig. \ref{fig:30bus}(a), we show that the total generation costs grows faster as load becomes larger, while the total carbon emissions grows relatively slower. While in Fig. \ref{fig:30bus}(b), we observe the similar trend of 30-bus and 6-bus system with regard to the average and marginal emission rates. An interesting observation is, it has long been perceived that locational marginal price (LMP) will become higher along with the increase of system demand, while it does not hold for nodal marginal or average emission rates.

In all of our simulations, we validate our algorithm via  (i). for each node $i$, summing up nodal contributions of every generator excluding net power flow equal to the nodal demand $p_i^d$, so that demand's carbon emission can be traced back to perspective generators; and (ii). for each generator $j$, summing up all nodes' contributions made by $j$-th generator equal to the generation $p_j^g$. We find our proposed algorithm is computationally very efficient, using $0.011$ and $0.0022$ seconds on average for 30-bus and 6-bus simulation respectively on a 2018 Macbook Pro. This tool can be accompanied with power flow analysis, and add acceptable computation for carbon-related operations (e.g., carbon accounting) within short decision time windows.
It is also worthwhile to mention one byproduct of our approach is getting the exact carbon emission contribution made by single generator and line flow. Such information will be also helpful for decision makers and policymakers in finding carbon reduction strategies, like reducing carbon emissions for targeted generators or particular regions.

\section{Conclusions and Future Work}
In this paper, an efficient framework is proposed to quantify the nodal average carbon emission and marginal carbon emission rates in a power network. Our method is based on calculating each generator's contribution to the nodal demand by depth-first tree search with each generator set as tree's root node. The proposed approach works with power networks of any topology and parameters, and it can also provide useful information on every generator's transferred carbon emissions to each demand node due to power consumption.

In future work, there are several interesting and important topics related to carbon emission calculation and profiling. First, we will consider analyzing temporal trend of carbon emission under multi-period scenarios where both electricity demand and generation constraints are temporally correlated. %In such cases, it is useful to analyze the carbon emission trend associated with each load. 
Moreover, our current work assumes perfect information of power flow and generations. We will consider how to recover the generators' working conditions and associated carbon emission rates, and also work with practical electricity market data to reveal real-world nodal carbon emissions. In addition, we will analyze the renewables' effects on the nodal carbon emission profile, and find grid planning solutions that better accommodate our proposed carbon evaluation metrics.

\bibliographystyle{ACM-Reference-Format}
\bibliography{sample-base}

\newpage
\appendix

\section{Optimal Power Flow Formulation}
\label{sec:opf}

In this section, we list one formulation of DCOPF, which is commonly used by independent system operators (ISO) for finding the power dispatch and line flow solution to minimize system costs~\cite{zhu2015optimization, chen2022learning}:

\begin{subequations}
\begin{align}
\min _{\mathbf{\theta}, \; \mathbf{p}^g} \quad & \mathbf{c}^T \mathbf{p}^g \\
\text { s.t. } \quad &  p^g_i- p^d_i=  \sum_{j:(i, j) \in \mathcal{L}} \beta_{i j}\left(\theta_i-\theta_j\right), & \forall i \in \mathcal{N} \label{opf:balance} \\
& \underline{p_{i j}} \leq- \beta_{i j}\left(\theta_i-\theta_j\right) \leq \overline{p_{i j}} , \quad  \forall(i, j) \in \mathcal{L} \label{opf:line} \\
& \underline{\mathbf{p}^g} \leq \mathbf{p}^g \leq \overline{\mathbf{p}^g} \label{opf:gen} \\
& \theta_{\text {ref }}=0 . 
\end{align}
\end{subequations}

In this model, $\mathbf{c}$ denotes the cost vector of generators. $\theta_i$ denotes the phase angle at node $i$, and we have each line flow $p_{ij}=\beta_{i j}\left(\theta_i-\theta_j\right)$. Constraint \eqref{opf:balance}-\eqref{opf:gen} denote the nodal power balance, the line flow limits, and the power generation limits respectively. Once solved, we can get the power generation $\mathbf{p}^g$ as well as line flow $\mathbf{p}_{line}$. Then we are able to implement our algorithm for nodal carbon emission rate calculation.

We note that our approach shall not be restricted to the linearized DCOPF problem or linear generation costs. As long as we know the line flow and generation value, we can implement our algorithm to find nodal carbon emissions. Our approach can be also extended to consider the power losses on lines by examining each generator's relationship with the line loss.

One notable characteristic in the optimal power flow model is the locational marginal price, which quantifies the price sensitivities of nodal demand. The locational marginal carbon emission rate can be understood as the LMP equivalent in terms of carbon emissions. In general, OPF gives solution with lower LMP when total network load is small, and the LMP becomes higher when there are line congestion or more expensive generators are used for satisfying the demands. But on the contrary to the situation of LMP, as illustrated in Fig. \ref{fig:30bus}, locational marginal emission rate is in general decreasing when net load increases. More interestingly, we observe that when some node increases demand, some other nodes' average carbon emission rate $\delta(p_i^d)$ even decreases. This is due to the fact that the more expensive generation such as gas-fueled generators usually have lower emission rate $\gamma_k$.

\section{Load Regions and Locational Marginal Carbon Emissions}
\label{sec:region}

\subsubsection*{Load Regions with Same Generation Policy} It has been shown in literature~\cite{ji2016probabilistic} that for DCOPF problem, there exists a neighborhood $\mathcal{N}(\mathbf{p}^d)$ around load $\mathbf{p}^d$, where the optimal generation policy $\mathbf{p}^g=f^*(\mathbf{p}^d)$ holds the same. More importantly, within $\mathcal{N}(\mathbf{p}^d)$, the marginal generators and line flow patterns are  the same, making the dispatch rule unchanged. Thus as long as the small perturbation $\epsilon$ introduced in Sec \ref{sec:marginal} does not change the load region, marginal carbon emission rate computation via perturbation method is exact.

\subsubsection*{Relationship to Power Transfer Distribution Factor} In power flow studies, the Power Transfer Distribution Factors (PTDF) indicate the relationship between line flow and nodal power injections. Specifically, denote $\mathbf{F} \in \mathbb{R}^{N\times L}$ to be the power flow distribution factor matrix, where $\mathbf{F}_{il}$ determines how a power injection at node $i \in [1, N]$ affects power flow across line $l \in [1,L]$. In the literature, there is also formulation of OPF problem using PTDF matrix to replace nodal power balance constraint \eqref{opf:balance} and line flow limit \eqref{opf:line} with $\underline{\mathbf{p}}_{line} \leq \mathbf{F}(\mathbf{p}^g-\mathbf{p}^d)\leq \overline{\mathbf{p}}_{line}$.

Moreover, it is interesting to note that by physical interpretations, PTDF matrix can be used to calculate the mapping between net power injection and line flow. However, there are two intrinsic issues associated with using PTDF matrix to calculate the line flow contributions from each generator. If generators' participation factor changes, the PTDF matrix is different, and needs to go through the whole process to recalculate $F$. Another issue is to calculate PTDF, matrix inversion is usually necessary, while under cases such as a group of generators are connected to the rest of the system through a single node, it is not invertible. We refer to \cite{cheng2005ptdf} for a more detailed discussions on the PTDF formulation and applicable scenarios. Moreover, once PTDF matrix is derived, it is still necessary to calculate nodal emission rates $\mu_i$ or $\delta_i$ based on nodal line flow injections and generation injections. Our proposed method can be applied to the general case as long as the power network is a connected graph. Moreover, for any power flow sample, we can quickly solve each generator's carbon emission contribution without going through the calculation process of PTDF and matrix inversion.

\section{Discussion on Cycles}
\label{sec:cycle}
In the algorithm described in Sec. \ref{sec:method}, we note that when there is no directed cycle in the reachable set of the generator $k$, it is then valid to implement the depth-first tree search with termination conditions. Here we show that for DC power flow model, such directed cycle does not exist.

\begin{proposition}
Given network demand $\mathbf{p}^d$ and generation $\mathbf{p}^g$, the reachable set $\mathcal{S}_g$ of $p_k^g$ does not contain directed cycle.
\end{proposition}

\begin{proof}
We prove it by contradiction. Suppose there exists a directed path in the network which forms a loop of size $N_c$, e.g., $\{$bus 1, bus 2, ..., bus $N_c$, bus 1$\}$,  we have 
    \begin{equation}
        \begin{aligned}
 &\left(\theta_1-\theta_2\right)+\left(\theta_2-\theta_3\right)+\cdots+\left(\theta_{N_c}-\theta_1\right) \\
 =&\frac{p_{12}}{b_{12}}+\frac{p_{23}}{b_{23}}+\cdots+\frac{p_{n_c 1}}{b_{N_c 1}} \\
= & 0 .
\end{aligned}
\end{equation}
However, since all $b_{ij}$ and directed flow $p_{ij}$ are non-negative, the summation $\frac{p_{12}}{b_{12}}+\frac{p_{23}}{b_{23}}+\cdots+\frac{p_{n_c 1}}{b_{N_c 1}} >0$, which cannot be zero.
% all the line flows in  cannot take the same signs. \todo{didn't quite follow the proof... so $p_{ij} $ have different signs and all $b_{ij}$ are non-negative, why cannot the summation to be 0?}
Thus there is no directed cycle in the network flow.
\end{proof}

It is worth noting that our method can work well with standard power dispatch and operating procedures. In full AC power flow model, cycles may exist, and future steps on cycle eliminations or reductions of cycle flow may be needed for analysis.

\section{Variable Table and Simulation Settings}
\label{sec:setting}
In the following table, we collect the notations and definitions of  variables used in the paper.
\begin{center}
\begin{tabular}{c c } 
 \hline
 Notations  & Definition \\ [0.5ex] 
 \hline
  $p_{ij}(p_k^g)$ &  \makecell{Power flow on the line between node $i$ and \\node $j$ due to generator $k$ }\\ 
  $p_{i}^d (p_k^g)$ &\makecell{Generator $k$'s contribution to supply load $p_i^d$}\\
  $e(p_i^d)$ &\makecell{Total nodal carbon emission flow rate}\\
  $\gamma_k$ & \makecell{Generator $k$'s carbon emission rate}\\
  $\delta(p_i^d)$ & \makecell{Node $i$'s average carbon emission rate at load $p_i^d$}\\
  $\mu (p_i^d)$ & \makecell{Node $i$'s marginal carbon emission rate at load  $p_i^d$}\\
  $\mathcal{S}_k$ & \makecell{Reachable set of generator $k$}\\
 \hline
\end{tabular}
\end{center}

\begin{figure}[]
\centering
\includegraphics[width=0.45\textwidth]{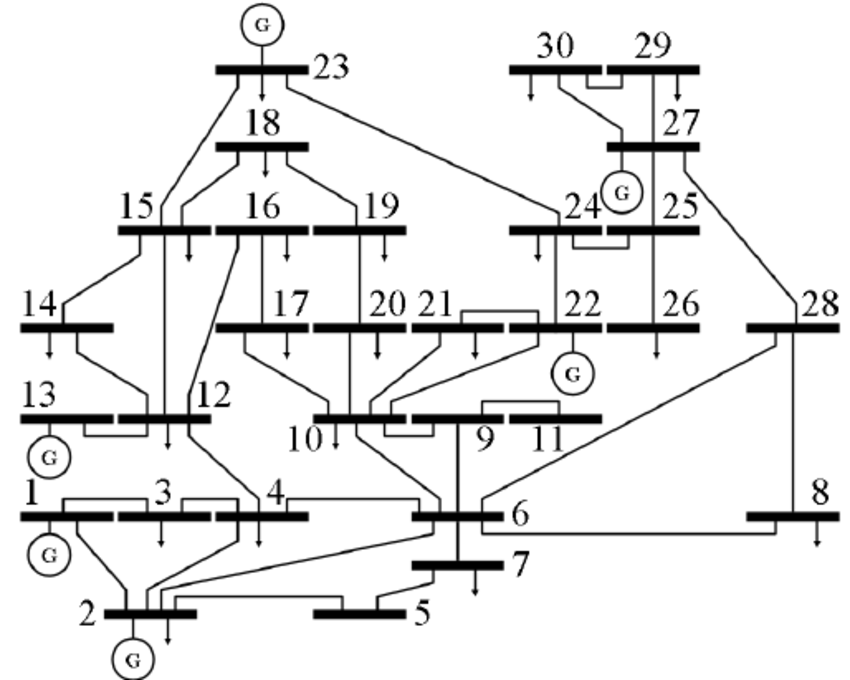}
\caption{IEEE 30-bus system test case.}
\label{fig:30bus_map}
\end{figure}

In the following, we describe the settings for the OPF problem used for solving power dispatch of 6-bus and 30-bus system. 

\textbf{6-bus system}: The generators are located at bus 1, bus 3, and bus 6,  with cost vector $\mathbf{c}=[100, 150, 240]^T$ pew MW, and carbon emission rate vector $\mathbf{\gamma}=[2388, 904, 321]^T$.

\textbf{30-bus system}: The generators are located at bus 1, bus 2, bus 13, bus 22, bus 23, and bus 27,with cost vector $\mathbf{c}=[100, 150, 240, 350, 500, 300]^T$ per MW, carbon emission rate vector $\mathbf{\gamma}=[2159, 2002, 1611, 890, 577, 113]^T$. The system topology map is shown in Fig. \ref{fig:30bus_map}.

Note that in both systems, the general trend is generator with lower costs have higher marginal emission rate, which also generally holds for fossil fuel generations in power networks. We note that we did not simulate renewable generations, as they can be firstly absorbed into net load and then solve the OPF problem with only fossil fuel generators involved.

\end{document}